\pdfoutput=1
\documentclass{article}

\usepackage{arxiv}

\usepackage{cite}
\usepackage{amsmath,amssymb,amsfonts}
\usepackage{algpseudocode}
\usepackage{graphicx}
\usepackage{textcomp}
\usepackage{xcolor}
\usepackage{multirow}
\usepackage[utf8]{inputenc} 
\usepackage[T1]{fontenc}    
\usepackage{hyperref}       
\usepackage{url}            
\usepackage{amsthm}
\usepackage[ruled,noline]{algorithm2e}
\usepackage{mathtools}

\DeclarePairedDelimiter{\ceil}{\lceil}{\rceil}

\newcommand{\BO}[1]{{O}\left(#1\right)}
\newcommand{\BT}[1]{{\Theta}\left(#1\right)}
\newtheorem{theorem}{Theorem}

\newtheorem{lemma}{Lemma}

\providecommand{\keywords}[1]
{
  \small	
  \textbf{\textit{Keywords---}} #1
}

\begin{document}

\title{Scalable Distributed Approximation of Internal Measures for Clustering Evaluation}

\author{Federico Altieri\\
University of Padova\\
 Padova,  Italy\\
\texttt{altieri@dei.unipd.it}\\
\And
Andrea Pietracaprina\\
University of Padova\\
 Padova,  Italy\\
\texttt{capri@dei.unipd.it}\\
\And
Geppino Pucci\\
University of Padova\\
 Padova,  Italy\\
\texttt{geppo@dei.unipd.it}\\
\And
Fabio Vandin\\
University of Padova\\
 Padova,  Italy\\
\texttt{vandinfa@dei.unipd.it}
}

\date{}

\maketitle




\begin{abstract} \small\baselineskip=9pt 
An important step in cluster analysis is the evaluation of the quality
of a given clustering through structural measures of goodness.
Measures that do not require additional information for their
evaluation (but the clustering itself), called internal measures, are
commonly used because of their generality. The most widely used
internal measure is the \textit{silhouette coefficient}, whose na\"ive
computation requires a quadratic number of distance calculations,
unfeasible for massive datasets. Surprisingly, there are no known
general methods to efficiently approximate the silhouette coefficient
of a clustering with rigorously provable high accuracy. In this paper,
we present the first scalable algorithm to compute such a rigorous
approximation for the evaluation of clusterings based on any metric
distances. Our algorithm approximates the silhouette coefficient
within a mere additive error $\BO{\varepsilon}$ with probability
$1-\delta$ using a very small number of distance calculations, for any
fixed $\varepsilon, \delta \in (0,1)$.  We also provide a distributed
implementation of the algorithm using the MapReduce model, which runs
in constant rounds and requires only sublinear local space at each
worker, thus making our estimation approach applicable to big data
scenarios.  An extensive experimental evaluation provides evidence
that our algorithm returns highly accurate silhouette estimates,
unlike competing heuristics, while running in a fraction of the time
required by the exact computation.
\end{abstract}


\section{Introduction}
\label{sec:intro}

Clustering is a fundamental primitive for the unsupervised analysis of datasets, and finds applications in a number of areas including pattern recognition, bioinformatics, biomedicine, and data management~\cite{AggarwalR13}. In its more general definition, clustering requires to identify groups of elements where each group exhibits high similarity among its members, while elements in different groups are dissimilar. 
Starting from this common definition, 
several algorithms have been proposed to identify clusters in a dataset~\cite{HennigMMR15}, often formalizing clustering as an optimization problem (based on a cost function). The resulting optimization problems are usually computationally hard to solve, and algorithms providing rigorous approximations are often sought in such cases. More recently, the focus has been on developing efficient methods that scale to the massive size of modern datasets~\cite{AwasthiB15,MalkomesKCWM15,BargerF16,CeccarelloFPPV17,CeccarelloPP19,MazzettoPP19} while still providing rigorous guarantees on the quality of the solution.

A common step after clustering has been performed is \emph{clustering evaluation} (sometimes called clustering \emph{validation}). Clustering evaluation usually employs an evaluation measure capturing the goodness of the structure of a clustering . Evaluation measures are classified into \emph{unsupervised} or \emph{internal} measures, which do not rely on external information, and \emph{supervised} or \emph{external} measures, which assess how well a clustering matches the structure defined by external information \cite{TanSK06}. While external measures are useful only when additional external knowledge regarding the cluster structure of the data is available, internal measures find applications in every scenario.

The most commonly used internal measure for clustering evaluation is
the \emph{silhouette coefficient}~\cite{Rousseeuw1987} (for brevity,
called \emph{silhouette} in this paper). 
The silhouette of a
clustering is the average silhouette of all elements in the clusters,
and, in turn, the silhouette $s(e)$ of an element $e$ in some cluster
$C$ is defined as the ratio $(b(e)-a(e))/\max\{a(e),b(e)\}$, where
$a(e)$ is the average distance of $e$ from the other elements of $C$,
and $b(e)$ is the minimum average distance of $e$ from the elements of
another cluster $C'$. In other words, $s(e)$ provides and indication
to what extent $e$ is closer (on average) to elements in its cluster
$C$ than to elements in the ``closest'' cluster $C'\neq C$.
The use of the silhouette for clustering evaluation is suggested
by widely used data mining books~\cite{TanSK06,han2011data}, and has
found application in several important areas~\cite{hossain2007gdclust,
  ngefficient, sellam2016blaeu, wiwie2015comparing}.  The na\"ive
computation of the silhouette for a clustering of $n$ elements
requires $\BT{n^2}$ distance calculations, which is unfeasible for
massive dataset that require distributed clustering
solutions~\cite{ene2011fast,feldman2011unified,bahmani2012scalable,balcan2013distributed,MalkomesKCWM15}.
Surprisingly, while several methods have been proposed to efficiently
cluster large datasets with rigorous guarantees on the quality of the
solution, there are no methods to efficiently approximate the
silhouette featuring \emph{provably high accuracy}.

\subsection{Related Work}
\label{sec:related}
While the silhouette is one of the most popular internal measures for
clustering evaluation \cite{MoulaviJCZS14,XiongL15,TomasiniEBM16}, the quadratic
complexity of the na\"{i}ve exact calculation makes its use impractical
for clusterings of very large datasets. For this reason, some attempts
have been made to propose variants that are faster to compute, or to
simplify its calculation in some special cases.  

Hruschka et al.~\cite{Hruschka2004}  present the \emph{simplified
  silhouette} for the evaluation of clusterings obtained on the basis
of objective functions involving squared Euclidean distances (e.g.,
$k$-means clusterings \cite{AggarwalR13}).  The simplified silhouette
is a variant of the silhouette, where for each element $e$
in a cluster, the aforementioned quantities $a(e)$ and $b(e)$ are
redefined, respectively,  as the (squared) distance between $e$ and the
centroid of its cluster, and of the closest centroid of another
cluster. In this fashion, the complexity of the whole computation
reduces to $\BO{nk}$. 
While Hruschka et al.~\cite{Hruschka2004}
and Wang et al.~\cite{Wang2017}
provide empirical evidence that the simplified silhouette can be an
effective evaluation measure for clusterings returned by Lloyd's
algorithm \cite{Lloyd82}, there is no evidence of its effectiveness
for other types of clusterings (e.g, clusterings based on other
distance functions) and, moreover,
the discrepancy between the original silhouette  and the
simplified silhouette can grow very large.

Frahling and Sohler \cite{FrahlingS08} proposed an optimization heuristic for speeding-up the computation of the exact silhouette
for clusterings based on Euclidean distances. For each element $e$ of a cluster $C$, while the term $a(e)$ is computed according to its definition, in
an attempt to reduce the operations needed to compute the term $b(e)$,
the heuristic first determines the average distance $d(e,C')$ between
$e$ and the elements of the cluster $C' \neq C$, whose centroid is
closest to $e$, and then sets $b(e)=d(e,C')$ in case the distance
between $e$ and the centroid of any other cluster $C'' \not\in
\{C,C'\}$ is larger than or equal to $d(e,C')$, since in this case
there is no need to compute any other average distance $d(e, C'')$.
However, when this is not the case, $b(e)$ must be computed according
to its definition and
its worst case complexity remains clearly quadratic.

The Apache Spark programming framework\footnote{\url{https://spark.apache.org/}} provides
optimized methods for computing the silhouette
of clusterings under $d$-dimensional squared Euclidean distances and under
one formulation of cosine distance.  Indeed, in these specific cases,
simple algebra suffices to show that precomputing, for each of the $k$
clusters, a limited number of values dependent on the coordinates of
the cluster's points, yields a fully parallelizable
algorithm featuring $O(nkd)$ work.  However, using squared distance measures to compute
the silhouette tends to
push the measure closer to $1$ compared to linear distances, thus
amplifying positive and negative scores.

In our algorithm, we employ a \emph{Probability Proportional to Size} (PPS)
sampling scheme that samples each element with probability proportional to a ``size'' measure.
The use of PPS sampling has been pioneered in the context of distance query processing and successfully applied to computing closeness centralities in graphs~\cite{Chechik2015}.  In the context of clustering, PPS has been used to obtain samples whose clustering cost approximates the clustering cost of the whole dataset, with guarantees on the quality of the approximation~\cite{cohen2018clustering}. To the best of our knowledge, prior to our work, the use of PPS for efficient clustering evaluation had not been explored.

\subsection{Our Contributions}
\label{sec:contributions}
In this work, we target the problem of the efficient computation of an
accurate estimate of the silhouette of a given clustering under
general metric distances. 
In this regard, our contributions are:
\begin{itemize}
\item 
We develop the first efficient, sampling-based algorithm for
estimating the silhouette with provable approximation guarantees. 
For any fixed $\varepsilon, \delta \in
(0,1)$, our algorithm approximates the silhouette of a $k$-clu\-ste\-ring
of $n$ elements within an additive error
$4\varepsilon/(1-\varepsilon)$ with probability at least $1-\delta$,
using only $\BO{nk\varepsilon^{-2}\log (nk/\delta)}$ distance
computations, which constitutes a dramatic improvement over the
$\BT{n^2}$ distance computations required by the na\"ive exact
algorithm.
\item 
We provide a distributed implementation of our algorithm using the
Map-Reduce framework \cite{DeanG08,PietracaprinaPRSU12},
which runs in constant rounds, and requires only sublinear space at
each worker.
\item 
\sloppy
We perform an extensive suite of experiments on real and
synthetic datasets to assess the effectiveness and efficiency of
our algorithm, and to compare it with known approaches 
for fast silhouette computation/approximation.
The experiments show that, in practice, our algorithm provides
silhouette estimates featuring very low absolute error (less than $0.01$ in most
cases) with small variance ($< 10^{-3}$) using a very small fraction of
the distance computations required by the exact calculation. Moreover,
the estimates returned by our algorithm are far superior in quality to 
those returned by   a  na\"ive  yet natural heuristics based on uniform sampling, 
or by  the simplifed silhouette~\cite{Hruschka2004} under general distances.
We also show that the MapReduce implementation of our algorithm
enables the estimation of the silhouette  for clusterings
of massive datasets (e.g., 1 billion elements) for which the exact
computation is out of reach.
\end{itemize}

Our algorithm addresses the problematic issues posed by existing approaches by providing an efficiently computable and provably accurate approximation of the silhouette,  allowing the use of the silhouette coefficient for very large datasets, for which its exact computation is unfeasible, and without the need to recur to other efficiently computable surrogates (such as the simplified silhouette coefficient) which in many cases  provide estimates so divergent from the exact silhouette  that they may lead to very different conclusions in the aforementioned scenarios of application.

In addition, while previously known approaches to
efficiently compute or approximate the silhouette have been developed
for specific distance functions, namely squared Euclidean and cosine
distances, our algorithm provides provably accurate silhouette
estimations for clusterings based on \emph{any} metric distance. Finally,
our algorithm can be generalized to compute rigorous approximations of other 
internal measures, such as cohesion and separation.

\subsection{Organization of the paper}
The rest of the paper is structured as
follows. Section~\ref{sec:methods} contains the description of our
proposed strategy for silhouette estimation, its accuracy and
performance analysis, and the MapReduce implementation.
Section~\ref{sec:exp} reports the results of an extensive suite of
experiments performed to evaluate the effectiveness and scalability of
our silhouette estimation algorithm on synthetic and real
datasets. Section~\ref{sec:conclusions} concludes the paper with some
final remarks.
\section{Methods} 
\label{sec:methods}
Consider a metric space $U$ with distance function $d(\cdot,\cdot)$,
and let $V = \{e_1, \dots e_n\} \subseteq U$ be a dataset of $n$
elements. Let also $\mathcal{C} = \{C_1, \dots C_k\}$ be a
\emph{$k$-clustering} of $V$, that is, a partition of $V$ into $k$
disjoint non-empty subsets called \emph{clusters}. A popular measure
to assess clustering quality was introduced by Rousseeuw in 1987
\cite{Rousseeuw1987}. Specifically, 
the \textit{silhouette of an element} $e \in V$ belonging to some cluster $C$, is defined as
\[
s(e) = \frac{b(e)-a(e)}{\max\{a(e), b(e)\}},
\]
where 
\[
a(e) = \frac{\sum_{e' \in C}d(e, e')}{|C|-1},
\;\;\;\;\;\;
b(e) = \min_{C_j\neq C}\frac{\sum_{e' \in C_j}d(e, e')}{|C_j|}
\]
denote, respectively, the average distance of $e$ from the other
elements in its cluster, and the minimum average distance of $e$ from
the elements in some other cluster.  The \emph{silhouette of the
  clustering} $\mathcal{C}$ is the average silhouette of all the
elements of $V$, namely:
\begin{equation}
\label{eq:avgShil}
s(\mathcal{C}) = \frac{\sum_{i = 1}^{n}s(e_i)}{n}.
\end{equation}
From the above definitions it immediately follows that the values
$s(e)$, with $e \in V$, and $s(\mathcal{C})$ range in $[-1,1]$.  A
positive value for $s(e)$ implies that $e$ has been assigned to an
appropriate cluster, while a negative value implies that there might
be a better cluster where $e$ could have been placed. Therefore,
$s(e)$ can be interpreted as a measure of the quality of the
clustering from the perspective of element $e$.  In turn,
$s(\mathcal{C})$ provides a global measure of the quality of the whole
clustering, where a value closer to 1 indicates higher quality.  The
exact computation of $s(\mathcal{C})$ requires $\BO{n^2}$ distance
calculations, which is prohibitive when dealing with large
datasets. In the following subsection, we present a randomized
strategy to yield an estimate of $s(\mathcal{C})$, which is accurate
within a provable error bound, with sufficiently high probability, and
amenable to an efficient distributed computation.

\subsection{A Fast Algorithm for Silhouette Estimation}
\label{sec:alg}
Consider the estimation of the silhouette $s(\mathcal{C})$ for a
$k$-clustering $\mathcal{C} = \{C_1, \dots C_k\}$ of a set $V$ of $n$
elements from a metric space $U$.  For each $e \in V$ and $C_j \in
\mathcal{C}$, define
\[
W_{C_j}(e) = \sum_{e' \in C_j}d(e, e').
\]
\sloppy
Note that for an element $e$ of a cluster $C$,
the quantities
$a(e)$ and $b(e)$ in the definition of the silhouette $s(e)$ can be
rewritten as $a(e) = W_{C}(e)/(|C|-1)$ and $b(e) = \min_{C_j \neq
  C}\{W_{C_j}(e_i)/|C_j|\}$.  Building on this observation, our
approach to approximating $s(\mathcal{C})$ is based on estimating the
$W_{C_j}(e)$'s by exploiting the \emph{Probability Proportional to
  Size} (PPS) sampling strategy used in 
\cite{Chechik2015}. In particular, for each cluster $C_j \in
\mathcal{C}$, a suitable small sample $S_{C_j}$ is computed, and, for
every element $e$, the value of $W_{C_j}(e)$ is
approximated in terms of the distances between $e$ and the elements of
$S_{C_j}$, ensuring a user-defined error bound with high
probability. The elements of $S_{C_j}$ are chosen according to a
carefully designed probability distribution which privileges the
selection of elements that are distant from some suitable ``central''
element of the cluster. The details are given in what follows.

Consider a fixed error tolerance threshold $0 < \varepsilon < 1$ and a
probability $0 < \delta < 1$. Our algorithm, dubbed
\textsc{PPS-Silhouette}, consists of two \emph{steps}. In Step 1, for each cluster
$C_j \in \mathcal{C}$, the algorithm computes a sample $S_{C_j}$ of
expected size $t= \ceil{(c/2\varepsilon^2)\ln{(4nk/\delta)}}$
for a suitably chosen constant $c$, while in Step 2
the approximate values of the $W_{C_j}(e)$'s and, in
  turn, the approximations of the $s(e)$'s are computed through
  the $S_{C_j}$'s. More precisely, in Step 1, each cluster $C_j$ is
  processed independently.  If $|C_j| \leq t$, then $S_{C_j}$ is set
  to $C_j$, otherwise, \emph{Poisson sampling} is performed over $C_j$,
that is, each element $e \in C_j$ is included in $S_{C_j}$ independently with
  a suitable probability $p_e$. Probability $p_e$, for $e \in C_j$, is determined as follows:

\begin{itemize}
\item First, an initial sample $S^{(0)}_{C_j}$ is selected, again by 
Poisson sampling, where each $e \in C_j$ is included
in $S^{(0)}_{C_j}$ independently with probability 
$(2/|C_j|)\ln (2k/\delta)$. This initial sample will contain,
with sufficiently high probability,  the aforementioned ``central''
element of $C_j$, namely, one that is close to a majority of the elements of $C_j$, a property which is necessary to
enforce the quality of the final sample $S_{C_j}$;
\item For each $\bar{e} \in S^{(0)}_{C_j}$ the value $W_{C_j}(\bar{e})$
is computed, and, for each $e \in C_j$, $p_e$ is set to $\min\{1, t\gamma_e \}$,
where $\gamma_e$ is the maximum between $1/|C_j|$ (corresponding to
uniform sampling) and the maximum of the values $d(e,\bar{e})/W_{C_j}(\bar{e})$,
with $\bar{e} \in S^{(0)}_{C_j}$ (representing the relative contributions of
$e$ to the $W_{C_j}(\bar{e})$'s of the sample points). 
\end{itemize}
In Step 2, for each $e \in V$ and each cluster $C_j$,
the sample $S_{C_j}$ is used to compute the value
\[
\hat{W}_{C_j}(e) = \sum_{e' \in S_{C_j}}\frac{d(e, e')}{p_{e'}}
\]
which is an accurate estimator of $W_{C_j}(e)$,
as will be shown by the analysis. Once all these values have been
computed, then, for each $e \in V$ belonging to 
a cluster $C$ we compute the estimates
$\hat{a}(e) = \hat{W}_{C}(e)/(|C|-1)$ and 
$\hat{b}(e) = \min_{C_j \neq C} \{\hat{W}_{C_j}(e)/|C_j|\}$,
which are in turn used to estimate the silhouette  as
$\hat{s}(e) = 
(\hat{b}(e)-\hat{a}(e))/(\max\{\hat{a}(e), \hat{b}(e)\})$.
Finally, we estimate the silhouette of the whole clustering as
\begin{equation}
\label{eq:avgShilhat}
\hat{s}(\mathcal{C}) = \frac{\sum_{e \in V}\hat{s}(e)}{n}.
\end{equation}

\begin{algorithm}[htb] \label{code:algorithm}
\SetAlgoLined
\label{alg:SilhCalc}
\small
\textbf{Input}: clustering $\mathcal{C}=\{C_1,\dots,C_k\}$ of $V$, 
$\varepsilon, \delta \in (0,1)$ \\
\textbf{Output}: estimate $\hat{s}(\mathcal{C})$ of $s(\mathcal{C})$.

\vspace*{0.2cm} 
// Step 1: \emph{PPS sampling}\\[0.1cm]
$n \leftarrow |V|$\;
$t \leftarrow \ceil{\frac{c}{2\varepsilon^2}\ln{(4nk/\delta)}}$; 
// (\emph{expected sample size}) \\
\For{\upshape\textbf{each} cluster $C_j\in \mathcal{C}$} {
\lIf{$t \geq |C_j|$}{$S_{C_j} \leftarrow C_j$}
\Else {
$S^{(0)}_{C_j} \leftarrow$ Poisson sampling of $C_j$
with probability \\
\hspace*{0.9cm} $(2/|C_j|)\ln (2k/\delta)$\;
\lFor{\upshape\textbf{each} $\bar{e} \in S^{(0)}_{C_j}$} {
$W_{C_j}(\bar{e}) \leftarrow \sum_{e' \in C_j} d(\bar{e},e')$}
\For{\upshape\textbf{each} $e \in C_j$} {
$\gamma_e \leftarrow \max \{d(e,\bar{e})/W_{C_j}(\bar{e}) : \bar{e} \in S^{(0)}_{C_j}\}$\;
$\gamma_e \leftarrow \max \{\gamma_e,1/|C_j|\}$\;
$p_e \leftarrow \min\{1, t\gamma_e\}$\;
}
$S_{C_j} \leftarrow$ Poisson sampling of $C_j$ with probabilities 
\\ \hspace*{0.9cm} $\{p_e : e \in C_j\}$\;
}
}
\vspace*{0.2cm}
// Step 2: \emph{silhouette estimation} \\[0.1cm]
\For{\upshape\textbf{each} $e \in V$} {
Let $e$ belong to cluster $C$\;
\For{\upshape\textbf{each} cluster $C_j$} {
$\hat{W}_{C_j}(e) = \sum_{e' \in S_{C_j}} d(e, e')/p_{e'}$\;
}
$\hat{a}(e) \leftarrow \hat{W}_{C}(e)/(|C|-1)$\;
$\hat{b}(e) \leftarrow \min \{\hat{W}_{C_j}(e)/|C_j| : C_j \neq C\}$\; 
$\hat{s}(e) \leftarrow 
(\hat{b}(e)-\hat{a}(e))/(\max\{\hat{a}(e), \hat{b}(e)\})$\;
}
$\hat{s}(\mathcal{C}) \leftarrow \sum_{e \in V} \hat{s}(e)/n$\;
\Return $\hat{s}(\mathcal{C})$
\caption{\textsc{PPS-Silhouette}($V, \mathcal{C}, \varepsilon, \delta$)}
\end{algorithm}

\subsection{Analysis}
\label{sec:analysis}
In this section we show that, with probability $1-\delta$, the value
$\hat{s}(\mathcal{C})$ computed by \textsc{PPS-Silhouette}
approximates the true silhouette $s(\mathcal{C})$ within a small error
bound, expressed in terms of $\varepsilon$. The key ingredient
towards this goal, stated in the following theorem, is a probabilistic
upper bound on the relative error of the estimate $\hat{W}_{C_j}(e)$
with respect to true value $W_{C_j}(e)$.
\begin{theorem} 
\label{th:mainerror}
There is a suitable choice of the constant $c$ in the definition of
the expected sample size $t$ used by \textsc{PPS-Silhouette}, which
ensures that, with probability at least $1 -\delta$, for every element
$e$ and every cluster $C_j$, the estimate $\hat{W}_{C_j}(e)$  is such that
\[
\left|\frac{\hat{W}_{C_j}(e) - W_{C_j}(e)}{W_{C_j}(e)} \right| 
\leq \varepsilon.
\]
\end{theorem}
\begin{proof}
The proof mimics the argument devised in \cite{Chechik2015}. 
Recall that 
$t = \ceil{\frac{c}{2\varepsilon^2}\ln{(4nk/\delta)}}$. 
Consider
an arbitrary cluster $C_j$ with more than $t$ elements
(in the case $|C_j| \leq t$ the theorem follows trivially). For an element
$e \in C_j$, let $m(e)$ denote the median of the distances from $e$ to all other
elements of $C_j$. Element $e$ is called \emph{well positioned} if
$m(e) \leq 2\min_{e' \in C_j} m(e')$. It is easy to see that at least
half of the elements of $C_j$ are well positioned. Hence, the initial
random sample $S^{(0)}_{C_j}$ will contain a well positioned 
element with probability
at least $(1-(2/|C_j|)\ln (2k/\delta))^{|C_j|/2} \geq 1-\delta/(2k)$.  
An easy adaptation of the proof of
\cite[Lemma 12]{Chechik2015}, shows that if $S^{(0)}_{C_j}$ contains a well
positioned element and $c$ is a suitable constant, then the Possion sample
$S_{C_j}$ computed with the  probabilities derived from $S^{(0)}_{C_j}$ is such
that 
$|(\hat{W}_{C_j}(e) - W_{C_j}(e))/ W_{C_j}(e)| \leq \varepsilon$,
with probability at least $1-\delta/(2nk)$.  By the union bound, it
follows that the probability that there exists a cluster $C_j$ such
that the initial sample $S^{(0)}_{C_j}$ does not contain a well positioned
element is at most $k\delta/(2k)=\delta/2$. Also, by conditioning on
the fact that for all clusters $C_j$ the initial sample $S^{(0)}_{C_j}$
contains a well positioned node, by using again the union bound we
obtain that the probability that there exists an element $e$ and a
cluster $C_j$ for which $|(\hat{W}_{C_j}(e) -
W_{C_j}(e))/W_{C_j}(e)| >\varepsilon$ is at most $nk \delta/(2nk)
= \delta/2$, which concludes the proof.
\end{proof}

From now on, we assume that the relative error bound stated in
Theorem~\ref{th:mainerror} holds for every element $e \in V$ and every
cluster $C_j$, an event which we will refer to as \emph{event
  $E$}. Consider an arbitrary element $e$, and let $\hat{s}(e)$ be the
estimate of the silhouette $s(e)$ computed by
\textsc{PPS-Silhouette}. The following key technical lemma,
establishes a bound on the absolute
error of the estimate.

\begin{lemma}
\label{lem:singlesilh}
If event $E$ holds, then 
$|\hat{s}(e)-s(e)| \le 4\varepsilon/(1-\varepsilon)$.
\end{lemma}

\begin{proof}
We first show that, if event $E$ holds, the relative errors
$|(\hat{a}(e)-a(e))/a(e)|$  and $|(\hat{b}(e)-b(e))/b(e)|$ 
are both upper bounded by $\varepsilon$.
For what concerns the former error, 
the bound follows immediately from the definition of $\hat{a}(e)$ and
the relative error bound assumed for $\hat{W}_{C_j}(e)$.
For the latter, consider an arbitrary cluster $C_j$ and let $b_{C_j}(e) =
W_{C_j}(e)/|C_j|$ and $\hat{b}_{C_j}(e) = \hat{W}_{C_j}(e)/|C_j|$. 
Arguing as for the previous bound, we can easly show that 
$|(\hat{b}_{C_j}(e)-b_{C_j}(e))/b_{C_j}(e)| \leq \varepsilon$.
Recall that $b(e) = \min_{C_j \neq C} b_{C_j}(e)$ and that
$\hat{b}(e) = \min_{C_j \neq C} \hat{b}_{C_j}(e)$. 
Suppose that $b(e)=b_{C'}(e)$ and $\hat{b}(e)=\hat{b}_{C''}(e)$,
for some, possibly different, clusters $C'$ and $C''$. We have:
\[
(1-\varepsilon)b(e) = (1-\varepsilon)b_{C'}(e) \leq (1-\varepsilon)b_{C''}(e) 
\leq \hat{b}_{C''}(e) = \hat{b}(e)
\]
and
\[
\hat{b}(e) = \hat{b}_{C''}(e) 
\leq \hat{b}_{C'}(e) \leq (1+\varepsilon)b_{C'}(e)= (1+\varepsilon)b(e),
\]
and the desired bound follows.

Now we establish a bound on the relative error for the term appearing
in the denominator of the formula defining $\hat{s}(e)$. Define $m(e)
= \max\{a(e) , b(e)\}$ and $\hat{m}(e) =
\max\{\hat{a}(e),\hat{b}(e)\}$. We show that
$|(\hat{m}(e)-m(e))/m(e)| \leq \varepsilon$.
Suppose that $m(e) = a(e)$, hence $a(e) \geq b(e)$ (the case $m(e) = b(e)$ can be dealt with similarly).  If $\hat{m}(e) = \hat{a}(e)$, the bound
is immediate.
Instead, if $\hat{m}(e) = \hat{b}(e)$, hence $\hat{b}(e) \geq \hat{a}(e)$, the bound follows since
\begin{eqnarray*}
\hat{m}(e) = \hat{b}(e) & \leq &
(1+\varepsilon)b(e) \\
& \leq & (1+\varepsilon)a(e) = (1+\varepsilon)m(e),
\end{eqnarray*}
and 
\begin{eqnarray*}
\hat{m}(e) = \hat{b}(e) & \geq & \hat{a}(e) \\
& \geq & (1-\varepsilon)a(e) \\
& = & (1-\varepsilon)m(e).
\end{eqnarray*}
We are now ready to derive the desired bound on 
$|\hat{s}(e) - s(e)|$. By virtue of the relative error bounds established 
above, we have that
\[
\hat{s}(e) = \frac{\hat{b}(e)-\hat{a}(e)}{\hat{m}(e)} 
\leq
\frac{(1+\varepsilon)b(e)-(1-\varepsilon)a(e)}{(1-\varepsilon)m(e)}.
\]
Straighforward transformations show that
\[
\frac{(1+\varepsilon)b(e)-(1-\varepsilon)a(e)}{(1-\varepsilon)m(e)}
=
s(e)+\frac{2\varepsilon}{1-\varepsilon}
\left(s(e)+\frac{a(e)}{m(e)}\right)
\]
Since $s(e) \leq 1$ and $a(e) \leq m(e)$,
we get
\[
\hat{s}(e) \leq s(e)+\frac{4\varepsilon}{1-\varepsilon}.
\]
In a similar fashion, it follows that
\begin{eqnarray*}
\hat{s}(e) & = & 
\frac{\hat{b}(e)-\hat{a}(e)}{\hat{m}(e)} \\
& \geq & 
\frac{(1-\varepsilon)b(e)-(1+\varepsilon)a(e)}{(1+\varepsilon)m(e)} \\
& = &
s(e)-\frac{2\varepsilon}{1+\varepsilon}
\left(s(e)+\frac{a(e)}{m(e)}\right) \\
& \geq &
s(e)-\frac{4\varepsilon}{1+\varepsilon}
\geq 
s(e)-\frac{4\varepsilon}{1-\varepsilon}.
\end{eqnarray*}
\end{proof}

An upper bound to the absolute error incurred when 
estimating $s(\mathcal{C})$ through the value $\hat{s}(\mathcal{C})$
computed by
\textsc{PPS-Silhouette}, is established in the following theorem,
whose proof is an immediate consequence of the definition
of the two quantities (Equations~\ref{eq:avgShil} and~\ref{eq:avgShilhat}),
and of Theorem~\ref{th:mainerror} and Lemma~\ref{lem:singlesilh}.
\begin{theorem}
\label{th:allsilh}
Let $V$ be a dataset of $n$ elements, and let
$\mathcal{C}$ be a $k$-clustering of $V$. Let $\hat{s}(\mathcal{C})$
be the estimate of the silhouette of the clustering $s(\mathcal{C})$
computed by \textsc{PPS-Silhouette} for given parameters
$\varepsilon$ and $\delta$, with $0 < \varepsilon, \delta < 1$, and
for a suitable choice of constant $c>0$ in the definition of the
sample size. Then,
\[
|\hat{s}(\mathcal{C})-s(\mathcal{C})|
\leq \frac{4\varepsilon}{1-\varepsilon}
\]
with probability at least $1-\delta$.
\end{theorem}

We now analyze the running time of \textsc{PPS-Silhouette}, assuming
that the distance between two points can be computed in constant
time. In Step 1, the running time is dominated by the computation of
the distances between the points of each sufficiently large cluster
$C_j$ and the points that form the initial sample $S^{(0)}_{C_j}$. A
simple application of the Chernoff bound shows that, with high
probability, from each such cluster $C_j$, $\BO{\log (nk/\delta)}$
points are included in $S^j_0$. Thus, Step 1 performs $\BO{n \log
  (nk/\delta)}$ distance computations, altogether.  For what concerns
Step 2, its running time is dominated by the computation of the
distances between all points of $V$ and the points of the union of all
samples. A simple adaptation of the proof
of \cite[Corollary 11]{Chechik2015} and a straightforward application
of the Chernoff bound shows that there are $\BO{kt}$ sample points
overall, with high probability, where $t$ is the expected sample size for each cluster. Recalling that $t = \lceil (c/2\varepsilon^2)\ln{(4nk/\delta)} \rceil$,
we have that Step 2 performs 
$\BO{nkt} =  \BO{(nk\varepsilon^{-2}) \log (nk/\delta)}$
distance computations overall. As a consequence, the running time of the algorithm is $\BO{nk\varepsilon^{-2}\log (nk/\delta)}$ which, for reasonable values of $k$, $\varepsilon$ and $\delta$, is a substantial improvement compared to the quadratic complexity of the exact computation.  

We remark that while the $\varepsilon^{-2}$ factor in the upper bound on the sample size might result in very large samples to achieve high accuracy, this quadratic dependence on $\varepsilon$ reflects a worst-case scenario and is needed to obtain the analytical bound. In fact, Section~\ref{sec:exp} will provide experimental evidence that, in practice, rather modest sample sizes are sufficient to achieve high levels of accuracy.

\subsection{Generalization To Other Measures}
\label{sec:gen}
The PPS-based sampling strategy adopted for approximating the
silhouette of a clustering can be applied to other measures used for
internal clustering evaluation, which are based on sums of distances
between points of the dataset. This is the case, for instance, of
measures that compare the \emph{cohesion} (i.e., the average
intracluster distance) and the \emph{separation} (i.e., the average
intercluster distance) of the clustering \cite{TanSK06}.

More precisely, consider a $k$-clustering $\mathcal{C} = \{C_1, \dots
C_k\}$ and define the cohesion and separation of $\mathcal{C}$ as
\begin{eqnarray*} 
\mbox{Coh}(\mathcal{C}) & = &
\frac{1}{2} 
\frac{\sum_{j=1}^k \sum_{e',e'' \in C_j} d(e',e'')}
{\sum_{j=1}^k {|C_j| \choose 2}} \\ 
\mbox{Sep}(\mathcal{C}) & = &
\frac{\sum_{1 \leq j_1 < j_2 \leq k}\sum_{e' \in C_{j_1}}\sum_{e'' \in C_{j_2}} 
d(e',e'')}
{\sum_{1 \leq j_1 < j_2 \leq k} (|C_{j-1}||C_{j_2}|)},  
\end{eqnarray*}
respectively. (These measures have been also employed to assess the
average cluster reliability for a clustering of a network where
distances correspond to connection probabilities
\cite{LiuJAS12,CeccarelloFPPV17}.)  We can rewrite the above measures
in terms of the sums $W_{C}(e)$ defined in the previous section, as
follows
\begin{eqnarray*}
\mbox{Coh}(\mathcal{C}) & = &
\frac{1}{2}
\frac{\sum_{j=1}^k \sum_{e \in C_j} W_{C_j}(e)}
{\sum_{j=1}^k {|C_j| \choose 2}} \\      
\mbox{Sep}(\mathcal{C}) & = &
\frac{\sum_{1 \leq j_1 < j_2 \leq k}\sum_{e \in C_{j_1}} W_{C_{j_2}}(e)}
{\sum_{1 \leq j_1 < j_2 \leq k} (|C_{j_1}||C_{j_2}|)} \\   
\end{eqnarray*}
Clearly, approximations of the $W_{C}(e)$'s with low relative errors
immediately yield approximations with low relative error for
$\mbox{Coh}(\mathcal{C})$ and $\mbox{Sep}(\mathcal{C})$. Specifically,
define $\widehat{\rm Coh}(\mathcal{C})$ and $\widehat{\rm
  Sep}(\mathcal{C})$ as the respective approximations to
$\mbox{Coh}(\mathcal{C})$ and $\mbox{Sep}(\mathcal{C})$ obtained by
substituting each $W_{C}(e)$ occurring in the above equations 
with the value $\hat{W}_{C}(e)$ computed within
\textsc{PPS-Silhouette}. The following theorem is an immediate
consequence of Theorem 2.1.
\begin{theorem}
\label{th:cohsep}
Let $V$ be a dataset of $n$ elements, and let
$\mathcal{C}$ be a $k$-clustering of $V$. Let $\widehat{\rm
  Coh}(\mathcal{C})$ and $\widehat{\rm Sep}(\mathcal{C})$ be the
respective approximations to $\mbox{Coh}(\mathcal{C})$ and
$\mbox{Sep}(\mathcal{C})$ based on the values $\hat{W}_{C}(e)$
computed within \textsc{PPS-Silhouette} for given parameters
$\varepsilon$ and $\delta$, with $0 < \varepsilon, \delta < 1$, and
for a suitable choice of constant $c>0$ in the definition of the
sample size. Then,
\begin{eqnarray*}
\left|
\frac{\widehat{\rm Coh}(\mathcal{C}) - {\rm Coh}(\mathcal{C})}
{{\rm Coh}(\mathcal{C})}
\right| & \leq & \varepsilon \\
\left|
\frac{\widehat{\rm Sep}(\mathcal{C}) - {\rm Sep}(\mathcal{C})}
{{\rm Sep}(\mathcal{C})}
\right| & \leq & \varepsilon \\
\end{eqnarray*}
with probability at least $1-\delta$.
\end{theorem}

\subsection{Map-Reduce Implementation}
\label{sec:MR}

Throughout the algorithm, a key-value pair will be written as
$(\mbox{\it key} \; | \; \mbox{\it value}_1, \ldots, \mbox{\it
  value}_r)$, where $\langle \mbox{\it value}_1, \ldots, \mbox{\it
  value}_r \rangle$ is an $r$-dimensional value. Consider a
$k$-clustering $\mathcal{C} = \{C_1, \dots, C_k\}$ of a dataset $V =
\{e_1, \dots e_n\}$ of $n$ elements.  We assume that, initially, the
clustering is represented by the following set of key-value pairs:
\[
\{ (i | e_i, j_i) : 1 \leq i \leq n \wedge e_i \in C_{j_i} \}.
\]
We also assume that the values $n$, $k$, $t$ used in
\textsc{PPS-Silhouette}, as well as all $|C_j|$'s, are given to the
program.  Recall that our implementation, outlined in the main paper,
partitions the points of $V$ into $w$ subsets of size $n/w$ each,
where $w \in [0,n-1]$ is a suitable design parameter.  The four rounds
are described in detail below.  As customary in the description of
MapReduce algorithms, we will use the term \emph{reducer} to refer to
the application of a reduce function to a set of pairs with same
key. \\

\noindent
{\bf Round 1:}
\begin{itemize}
\item{\textbf{Map}}:
Map each pair $(i | e_i, j_i)$ into the pair $(i \bmod w |
e_i,j_i,0)$. Also, with probability $(2/|C_{j_i}|)\ln (2k/\delta)$
select  $e_i$ to be part of the initial Poisson sample
$S^{(0)}_{C_{j_i}}$, and produce the $w$ additional pairs $(\ell | e_i,j_i,1)$,
with $0 \leq \ell < w$.
\item{\textbf{Reduce}}:
Let $V_{\ell}$ be the set of elements $e_i$ for which there is 
a pair $(\ell | e_i,j_i,0)$. (Observe that the $V_{\ell}$'s form a
balanced partition of $V$ into $w$ subsets.)
For each pair $(\ell | e_i,j_i,1)$
compute the sum $W_{i,\ell}$ 
of the distances from $e_i$ to all elements of $V_{\ell} \cap C_{j_i}$
and produce the pair $(\ell | e_i,j_i,W_{i,\ell},1)$.
Instead, for each pair $(\ell | e_i,j_i,0)$
produce the pair $(\ell | e_i,j_i,0,0)$.
\end{itemize}
By the analysis (see main paper), we know that, with high
probability, $\BO{k \ln(nk/\delta)}$ points have been selected to be
part of the initial samples $S^{(0)}_{C_{j_i}}$'s (a copy of all these
points, represented by pairs with 1 at the end, exists for each key
$\ell$). Therefore, this round requires local memory $M_L = \BO{n/w+k
  \ln(nk/\delta)}$, with high probability, and aggregate memory $M_A =
\BO{wM_L}$. \\

\noindent
{\bf Round 2:}
\begin{itemize}
\item{\textbf{Map}}: 
Map each pair $(\ell | e_i,j_i,0,0)$ into itself, and
each pair  $(\ell | e_i,j_i,W_{i,\ell},1)$ into
the $w$ pairs $(\ell' | e_i,j_i,W_{i,\ell},1)$,
with $0 \leq \ell' < w$. 
\item{\textbf{Reduce}}: 
Each reducer, corresponding to some key $\ell$, now contains all
the information needed to compute the values $W_C(e)$ for each cluster $C$
and element $e \in S^{(0)}_C$. This in turn implies that the reducer
can also compute all sampling probabilities $p(e_i)$ for $e_i \in V_{\ell}$,
by executing the operations specified in \textsc{PPS-Silhouette}.
The reducer computes these probabilities and produces the pairs $(\ell | e_i,j_i,p(e_i))$, one for each $e_i \in V_{\ell}$.
\end{itemize}
By the observations made at the end of Round 1, this round requires
local memory $M_L = \BO{n/w+wk \log(nk/\delta)}$, with high
probability, and aggregate memory $M_A = \BO{wM_L}$. \\

\noindent
{\bf Round 3:}
\begin{itemize}
\item{\textbf{Map}}: 
Map each pair $(\ell | e_i,j_i,p(e_i))$ into the pair 
$(\ell | e_i,j_i,0)$. Also, with probability $p(e_i)$
select $e_i$ to be part of the Poisson sample
$S_{C_{j_i}}$, and produce the $w$ additional pairs $(\ell' | e_i,j_i,1)$,
with $0 \leq \ell' < w$.
\item{\textbf{Reduce Phase}}: 
Each reducer, corresponding to some key $\ell$, now contains all
the information needed to compute the approximate silhouette
values $\hat{s}(e_i)$ for each element in $e_i \in V_{\ell}$ and then sum them into
$\hat{s}_{\ell}$, producing the single pair $(0 | \hat{s}_{\ell})$.
\end{itemize}
By the analysis (see main paper), we know that,
with high probability, the size of union of the Poisson samples
$S_{C_j}$'s is $\BO{kt}$, with
$t = \lceil c/(2\varepsilon^2) \ln(4nk/\delta) \rceil$.
Therefore, this round requires local memory $M_L =
\BO{n/w+kt} = \BO{n/w+(k/\varepsilon^2)\log (nk/\delta)}$,
with high probability, 
and aggregate memory $M_A = \BO{wM_L}$. \\

\noindent
{\bf Round 4:}
\begin{itemize}
\item{\textbf{Map}}: identity map.
\item{\textbf{Reduce}}: 
The reducer corresponding to key 0 computes and returns 
$\hat{s}_{\cal C} = (1/n)\sum_{\ell = 0}^{w-1} \hat{s}_{\ell}$.
\end{itemize}
This round requires $M_L = M_A = \BO{w}$. 

Overall, the above 4-round MapReduce algorithm requires local memory
$M_L = \BO{n/w+(w+1/\varepsilon^2) k \log(nk/\delta)}$, with high
probability, and aggregate memory $M_A = \BO{wM_L}$. 

\section{Experimental Evaluation}
\label{sec:exp}
We ran extensive experiments with the twofold objective of assessing
the quality and evaluating the performance of our
\textsc{PPS-Silhouette} algorithm.  Concerning quality assessment, we
evaluate the accuracy of the estimate provided by
\textsc{PPS-Silhouette}, and compare it against known heuristic
and exact (specialized) baselines.
 For what concerns performance, we evaluate the scalability
of the MapReduce implementation of \textsc{PPS-Silhouette} and compare
its performance against the one of the aforementioned baselines. 

\subsection{Baselines}
\label{sec:baselines}
We gauge the performance of our \textsc{PPS-Silhouette} algorithm
against five  \emph{baselines}. The first baseline is the exact
computation based on the definition. The second baseline 
is a variation of \textsc{PPS-Silhouette}, where the samples $S_{C_j}$ are
chosen via a uniform Poisson sampling, using the same probability
$p(e)=t/|C_j|$ for each $e \in C_j$.
 The other three baselines implement,
respectively, the simplified silhouette of
\cite{Hruschka2004,Wang2017}, the Frahling-Sohler optimization of the exact computation \cite{FrahlingS08}, and
the optimized exact method available in the Apache Spark library
for squared Euclidean distances
(see Section~\ref{sec:related}) 
Our quality assessment compares \textsc{PPS-Silhouette}
against the approximation baselines, namely 
uniform sampling and simplified silhouette, while our performance evaluation compares \textsc{PPS-Silhouette}
against all baselines except for simplified silhouette, since 
the latter turned out to have low accuracy.

\subsection{Implementation and Environment}
\label{sec:impEnv}
For quality assessment, we devised
Java-based sequential implementations of 
\textsc{PPS-Silhouette}, of the exact computation, based on the definition,
and of the uniform sampling and simplified silhouette baselines.
For performance evaluation, we devised MapReduce
implementations, using the Apache Spark programming framework with 
Java\footnote{\url{https://spark.apache.org}},  
of all baselines except simplified silhouette and
the Apache Spark optimized method, whose code is provided in
the Spark library.
The MapReduce implementation of the  uniform sampling baseline is
patterned after the one of \textsc{PPS-Silhouette}, by only changing the
sampling probabilities. The MapReduce implementation of the exact silhouette
computation aims at minimizing the computational load per worker
by  evenly partitioning the operations among the workers and by providing each worker with an entire copy of the dataset. Finally, the
MapReduce implementation of the Frahling-Sohler optimization 
enforces the optimistic scenario where the heuristics for the
$b(e)$ term, in the computation of $s(e)$, is always successful, in order 
to appreciate the maximum speed-up that can be achieved over the
non-optimized implementation\footnote{All our implementations are available at \url{https://github.com/CalebTheGame/AppxSilhouette}}.


In our experiments we fixed
$\delta=0.1$, resulting in a 90\% confidence on the approximation
guarantee, and we opted to control accuracy by varying
the composite parameter $t$ directly
rather than fixing the independent variables
$\varepsilon$ and $c$ separately (see Section~\ref{sec:analysis}).

The experiments were run on 
CloudVeneto\footnote{\url{http://cloudveneto.it}}, an institutional
cloud, which provided us a cluster of
9 PowerEdge M620 nodes, each with octa-core Intel Xeon E5-2670v2 2.50
GHz and 16 GB of RAM, running Ubuntu 16.04.3 LTS with Hadoop 2.7.4 and
Apache Spark 2.4.4. 

\begin{table*}[h!]
\caption{Maximum and average absolute error of the  silhouette estimates returned by \protect\textsc{PPS-Silhouette} (label ``PPS'') and the uniform sampling algorithm, (label ``uniform'') relative to the $k$-medoid clustering of the synthetic dataset with $n=2\cdot 10^4$ points, $k=2,3 \ldots, 10$, and $t=64, 256$ and $1024$, alongside with absolute error of the simplified silhouette. Values $<10^{-3}$ are denoted by $0^*$.}
{\small
\begin{center}
\begin{tabular}{c|c|cc|cc|cc|cc|cc|cc}
\multirow{4}{*}{\textbf{k}} &\multicolumn{13}{c}{$\textbf{Absolute Error}$}\\
\cline{2-14}
&  &\multicolumn{6}{c|}{\textbf{PPS}}&\multicolumn{6}{c}{\textbf{Uniform}}\\
\cline{3-14}
& \textbf{Simplified} &\multicolumn{2}{c|}{\textbf{$t$=64}}&\multicolumn{2}{|c}{\textbf{$t$=256}}&\multicolumn{2}{|c|}{\textbf{$t$=1024}}&\multicolumn{2}{c|}{\textbf{$t$=64}}
&\multicolumn{2}{|c}{\textbf{$t$=256}}&\multicolumn{2}{|c}{\textbf{$t$=1024}}\\
\cline{3-14}
& \textbf{silhouette} &\textbf{max}&\textbf{avg}&\textbf{max}&\textbf{avg}&\textbf{max}&\textbf{avg}&\textbf{max}&\textbf{avg}&\textbf{max}&\textbf{avg}&\textbf{max}&\textbf{avg}\\
\hline
$2$&.283&.005&.001&.002&$0^*$&.001&$0^*$&.207&145&.206&.143&.199&.116\\
\hline
$3$&.404&.018&.005&.008&.002&.004&.001&.369&.235&.403&.217&.392&.163\\
\hline
$4$&.026&.012&.003&.005&.001&.002&$0^*$&.492&.291&.473&.301&.458&.150\\
\hline
$5$&.561&.015&.005&.008&.002&.004&.001&.174&.404&.487&.322&.438&.159\\
\hline
$6$&.314&.101&.015&.034&.007&.010&.002&.174&.082&.126&.049&.052&.005\\
\hline
$7$&.449&.047&.010&.016&.004&.010&.002&.280&.174&.273&.121&.097&.025\\
\hline
$8$&.321&.084&.017&.028&.006&.008&.002&.168&.075&.149&.349&.013&.002\\
\hline
$9$&.455&.078&.013&.015&.003&.003&$0^*$&.295&.187&.254&.188&.068&.002\\
\hline
$10$&.457&.050&.013&.013&.003&.003&$0^*$&.284&.180&.246&.108&.071&.002\\
\end{tabular}
\label{tab:errorSynth}
\end{center}
}
\end{table*}

\subsection{Datasets}
\label{sec:data}
We considered both synthetic 
and real datasets. Synthetic datasets have been chosen so to contain a
few outlier points, with the intent of making the accurate estimation
of the silhouette more challenging. Specifically, for different values
of $n$ (ranging from tens of thousands to one billion), we generated
$n-10$ points uniformly at random within the sphere of unit radius,
centered at the origin of the 3-dimensional Euclidean space, and $10$
random points on the surface of the concentric sphere of radius
$10^{4}$. For real datasets, we used reduced versions of the
``Covertype" and ``HIGGS"
datasets\footnote{\url{https://archive.ics.uci.edu/ml/datasets/}}.
The former contains 100000 55-dimensional points, corresponding to
tree observations from four areas of the Roosevelt National Forest in
Colorado; the latter contains 500000 7-dimensional points and is used
to train learning algorithms for high-energy Physics experiments (as
in previous works \cite{CeccarelloPP19,MalkomesKCWM15}, only the 7
summary attributes of the original 28 have been retained). For all
datasets, the clusterings used to test the algorithms have been
obtained by applying the $k$-medoids clustering algorithm implemented
in the Java Machine Learning
Library\footnote{\url{https://github.com/AbeelLab/javaml}}
\cite{abeel2009java}, using the Euclidean distance.

\subsection{Quality Assessment}
\label{subsec:qual}
In the first experiment,
we compared the accuracy of the silhouette
estimations returned by our \textsc{PPS-Silhouette} algorithm and by
the uniform sampling algorithm, using synthetic data. In order to make
the computation of the exact value feasible, we considered instances
of the synthetic dataset with a relatively small number of elements
($n=2\times 10^4$) and $k$-clusterings $\mathcal{C}$ with
$k\in\{2,\dots,10\}$. 
As explained in Section~\ref{sec:impEnv}, the accuracy of
\textsc{PPS-Silhouette}, and, clearly, also of the uniform sampling
algorithm, depends on the (expected) sample size $t$. Hence, to test
different levels of accuracy, for each value of $k$
we conveniently varied the value
$t$ directly. We remark that, in general, controlling the sample size $t$
directly, rather than indirectly through the several parameters
contributing to its definition, is more effective in practice since it
affords a better tailoring of the estimation to the available resources.
Specifically, we ran both \textsc{PPS-Silhouette} and the
uniform sampling algorithm with 
$t \in \{64, 128, 256, 512, 1024\}$.
As a measure of accuracy for the 
estimated silhouette $\hat{s}(\mathcal{C})$, we use the absolute error
$|\hat{s}(\mathcal{C}) - s(\mathcal{C})|$.
Table~\ref{tab:errorSynth} reports maximum and average of the absolute error over 100 runs of the two algorithms,
for each configuration of parameters $k$ and a selection of $t$'s
(the results for the other values of $t$, omitted for brevity, are
similar to the ones reported in the table).

The results show that \textsc{PPS-Silhouette} provides a very accurate
approximation to the silhouette already with $t=64$, for which the
average absolute error is at most 0.017 for all values of $k$ and the
maximum value is at most $0.084$ for all cases but $k=6$. The
approximation quickly improves when larger values of $t$ are
considered.
\textsc{PPS-Silhouette}
also features a very low variance, which is $<10^{-3}$ in all
cases. These results show that our \textsc{PPS-Silhouette} algorithm
provides a very accurate approximation of the silhouette even with a
limited number of samples. 

\begin{table*}[h!]
\caption{Comparison between exact silhouette, simplified silhouette, and Monte Carlo estimates obtained by averaging 100 estimates returned by \protect\textsc{PPS-Silhouette} (label ``PPS'') and the uniform sampling algorithm (label ``uniform''), for $k=2, 3, \ldots 10$ and $t=64, 128$ and $1024$. 
The highest value identified by each algorithm (one for each $t$ 
for PPS and uniform) is highlighted in boldface. Observe that the ``correct'' value of $k$ corresponds to the value with highest exact silhouette.}
{\small
\begin{center}
\begin{tabular}{c|c|c|ccc|ccc}
 &  \multicolumn{8}{c}{\textbf{Silhouette}} \\
\cline{2-9}
 \textbf{k}  &\textbf{Exact}&\textbf{Simplified}&\multicolumn{3}{c|}{\textbf{PPS}} & \multicolumn{3}{c}{\textbf{uniform}}\\
\cline{4-9}
 & & & \textbf{t=64}&\textbf{t=256}&\textbf{t=1024}&\textbf{t=64}&\textbf{t=256}&\textbf{t=1024}\\
\hline
2& 0.064&0.347& 0.064& 0.064& 0.064&\textbf{0.192}&\textbf{0.176}&0.074\\
3 & -0.072& 0.332& -0.070&-0.072&-0.072&0.151&0.081& -0.191 \\
4 & \textbf{0.343}&0.369&\textbf{0.343}&\textbf{0.344}&\textbf{0.343}&0.049&0.039& \textbf{0.183}\\
5 & -0.228& 0.333&-0.223&-0.228&-0.226&0.158&0.031&-0.244 \\
6 & 0.034& 0.348&0.022&0.032&0.034&0.121&0.071&0.040 \\
7 & -0.143& 0.306& -0.148&-0.145&-0.143&0.033&-0.030&-0.130\\
8 & 0.059&\textbf{0.380}&0.038&0.055&0.058&0.134&0.094&0.057 \\
9 & -0.087& 0.368& -0.102&-0.091&-0.088&0.102&0.021&-0.087\\
10 & -0.104& 0.353 & -0.116 &-0.017&-0.104&0.073&0.005&-0.103\\
\end{tabular}
\label{tab:expsilh1}
\end{center}
}
\end{table*}

On the other hand,
for $t=64$ the uniform sampling algorithm provides estimates with considerably larger
average and maximum absolute errors for most values of $k$,
with an average error that is is one order of magnitude larger than the average
error of \textsc{PPS-Silhouette}. The variance of the uniform sampling algorithm ($0.05$ on average) is 
also larger than the one of our algorithm.
These results show that the use of PPS
sampling is crucial to obtain good estimates of the silhouette
while keeping the sample size small.

The second column of Table~\ref{tab:errorSynth} shows the error that
occurs when the simplified silhouette is used to approximate the exact
value of the silhouette. Such error is always larger than twice the
maximum error provided by \textsc{PPS-Silhouette} (even for the
smallest values of the sample size $t$) and is approximately $0.36$ on
average, which is not surprising considering the weak relation between
simplified silhouette and the exact value of the silhouette quantified
by~\cite{Wang2017}.  Therefore, the simplified silhouette does not
provide accurate estimates of the exact value of the silhouette.

\begin{table*}[h!]
\caption{Comparison between exact silhouette, simplified silhouette, and maximum and average absolute error of the  silhouette estimates returned by \protect\textsc{PPS-Silhouette} (label ``PPS'') and the uniform sampling algorithm, (label ``uniform''), relative to the $k$-medoid clustering of the real datasets Covertype and HIGGS, with $k=5, 10$ and $t=64, 256$ and $1024$.}
{\small
\begin{center}
\begin{tabular}{c|c|c|c|cc|cc|cc|cc|cc|cc}
\multirow{3}{*}{\textbf{Dataset}}&\multirow{3}{*}{\textbf{k}}&\multicolumn{2}{c|}{\textbf{}}&\multicolumn{12}{c}{\textbf{Absolute error}}\\
\cline{5-16}
&& \multicolumn{2}{c|}{}&\multicolumn{6}{c}{\textbf{PPS}}&\multicolumn{6}{|c}{\textbf{Uniform}}\\
\cline{5-16}
& &  \multicolumn{2}{c|}{\textbf{Silhouette}}&\multicolumn{2}{c|}{\textbf{$t$ = 64}}&\multicolumn{2}{|c}{\textbf{$t$ = 256}}&\multicolumn{2}{|c|}{\textbf{$t$ = 1024}}&\multicolumn{2}{c|}{\textbf{$t$ = 64}}&\multicolumn{2}{|c}{\textbf{$t$ = 256}}&\multicolumn{2}{|c}{\textbf{$t$ = 1024}}\\
\cline{3-4}\cline{5-16}
 & & \textbf{Exact} & \textbf{Simpl.}&\textbf{max}&\textbf{avg}&\textbf{max}&\textbf{avg}&\textbf{max}&\textbf{avg}&\textbf{max}&\textbf{avg}&\textbf{max}&\textbf{avg}&\textbf{max}&\textbf{avg}\\
\hline
\multirow{2}{*}{Covertype}&$5$&0.344&0.474&.084&.016&.042&.006&.007&.002&.108&.019&.027&.007&.014&.003\\
\cline{2-16}
&$10$&0.266&0.394&.120&.023&.025&.006&.014&.002&.160&.031&.040&.009&.014&.002\\
\hline
\multirow{2}{*}{HIGGS}&$5$&0.244&0.368&.114&.016&.024&.006&.008&.003&.110&.017&.045&.008&.019&.003\\
\cline{2-16}
&$10$&0.151&0.289&.107&.028&.047&.010&.020&.004&.155&.033&.081&.012&.019&.005\\
\end{tabular}
\label{tab:errorReal}
\end{center}
}
\end{table*}

In the second experiment, we assessed the impact of relying on
estimated rather than exact values in the most typical scenario of use
of the silhouette, suggested in the original paper
\cite{Rousseeuw1987}, that is, identifying the best granularity $k$
when the same clustering algorithm is applied to a dataset with
different values of $k$. Specifically, we ran $k$-medoids clustering on
the synthetic dataset described above, for all values of $k$ in an
interval $[2, 10]$, and checked whether the best clustering (according
to the exact silhouette) was identified using the silhouette
estimations returned by \textsc{PPS-Silhouette} and by the uniform
sampling algorithm with expected sample size $t$, considering values
of $t \in \{64,256,1024\}$.  Since uniform sampling displays a large
variance affecting its silhouette estimates, we adopted a Monte Carlo
approach, where the average of 100 independent estimates is used in
lieu of a single estimate, to reduce the variance of the estimates. We
also computed the deterministic value of the simplified silhouette
\cite{Hruschka2004,Wang2017} to compare its effectiveness against the
two sampling strategies.  The results are shown in
Table~\ref{tab:expsilh1}, where the exact value of the silhouette is
also shown.

As expected, \textsc{PPS-Silhouette} always identifies the correct
value already for $t=64$, and, for all $k$ and $t$, it provides
extremely accurate approximations. On the other hand, the uniform
sampling algorithm identifies the correct value of $k$ only for
$t=1024$, and, for smaller values of $t$, it provides much weaker
approximations (indeed, $t=512$ is needed merely to hit the correct sign of
the silhouette value for all $k$).  Also, the table shows that the 
simplified silhouette is unable to identify the correct value of
$k$.

\begin{table*}[ht]
\caption{Identification of the value of $k$ yielding the highest silhouette using \protect\textsc{PPS-Silhouette} (label ``PPS'') and the uniform sampling algorithm (label ``uniform''). For each range of $k\in [2,\ell]$, with $\ell=3,\ldots, 10$ and  $t=64, 128, \ldots, 1024$, the table shows the percentage of the runs (over 100) where the correct value was identified.}
\begin{center}
\begin{tabular}{c|cc|cc|cc|cc|cc}
 & \multicolumn{10}{c}{\textbf{Correct choice (percentage)}} \\
\cline{2-11}
\textbf{Range} &\multicolumn{2}{c|}{\textbf{$t$ = 64}}&\multicolumn{2}{c|}{\textbf{$t$ = 128}}&\multicolumn{2}{c|}{\textbf{$t$ = 256}}&\multicolumn{2}{c|}{\textbf{$t$ = 512}}&\multicolumn{2}{c}{\textbf{$t$ = 1024}}\\
\cline{2-11}
 \textbf{of k} &\textbf{PPS}&\textbf{uniform}&\textbf{PPS}&\textbf{uniform}&\textbf{PPS}&\textbf{uniform}&\textbf{PPS}&\textbf{uniform}&\textbf{PPS}&\textbf{uniform}\\
\hline
2-3&100\%&66\%&100\%&56\%&100\%&57\%&100\%&71\%&100\%&85\%\\
2-4&100\%&1\%&100\%&2\%&100\%&3\%&100\%&31\%&100\%&72\%\\
2-5&100\%&0\%&100\%&0\%&100\%&6\%&100\%&25\%&100\%&64\%\\
2-6&100\%&0\%&100\%&2\%&100\%&4\%&100\%&28\%&100\%&50\%\\
2-7&100\%&2\%&100\%&1\%&100\%&5\%&100\%&20\%&100\%&53\%\\
2-8&100\%&3\%&100\%&2\%&100\%&5\%&100\%&13\%&100\%&53\%\\
2-9&100\%&2\%&100\%&1\%&100\%&7\%&100\%&25\%&100\%&53\%\\
2-10&100\%&1\%&100\%&3\%&100\%&5\%&100\%&19\%&100\%&57\%\\
\end{tabular}
\label{tab:expsilh2}
\end{center}
\end{table*}

In third experiment, we compared the accuracy of the silhouette
estimates computed by \textsc{PPS-Silhouette}, uniform sampling, and
simplified silhouette, on the real datasets HIGGS and Covertype,
considering $k=5, 10$ and $t = 64,256, 1024$.
Table~\ref{tab:errorReal} reports the the maximum and average absolute
errors of the silhouette estimates, over 100 runs, together with the
values of the exact silhouette and of the simplified
silhouette. Similarly to the case of the synthetic dataset,
\textsc{PPS-Silhouette} provides very accurate estimates already for
$t=64$, with an average error smaller than $0.03$, and for $t \geq
256$ it features an average error below $0.01$.  The variance is below
$10^{-3}$ in all cases. On these datasets, the uniform sampling
algorithm also provides fairly accurate estimates (with variance
$<10^{-3}$ in all cases). However, for no combination of $k$ and $t$
the average error of the uniform sampling algorithm is lower than the
one of \textsc{PPS-Silhouette}, while its maximum error is slightly
better only in three out of twelve configurations. Also, the
simplified silhouette displays poor accuracy, worse than the one of
\textsc{PPS-Silhouette} even with $t=64$.  Once again, this
experiment confirms the superiority of our \textsc{PPS-Silhouette}
algorithm with respect to the two competitors.


We ran an experiment similar in spirit to the second experiment, to further test the
capability of PPS versus uniform sampling of identifying the best
granularity $k$ when the same clustering algorithm is applied to a
dataset with different values of $k$.  Specifically, we ran $k$-medoids
clustering on the same synthetic dataset of the second experiment, for
all values of $k$ in an interval $[2, \ell]$, with $\ell \geq 3$, and
checked whether the best clustering (according to the exact
silhouette) was identified using the silhouette estimations returned
by \textsc{PPS-Silhouette} and by the uniform sampling algorithm. We
ran the experiment for $\ell = 3,\dots,10$, and expected sample sizes
$t \in \{64,128,256,512,1024\}$. For each combination of $\ell$ and
$t$, we considered 100 trials (corresponding to 100 estimates of the
silhouette by the two algorithms). Table~\ref{tab:expsilh2} displays
the fraction of times the correct value of $k$ was identified by the
two algorithms. Strikingly, \textsc{PPS-Silhouette} identified the
correct value of $k$ \emph{$100$\% of the times, for every combination
  of $k$ and $t$}, confirming the reliability and consistency of the
slihouette estimate even for small $t$ (e.g., $64$).

For what concerns the uniform sampling algorithm,
Table~\ref{tab:expsilh2} shows that, except for the range $k \in
\{2,3\}$ or $t=1024$, the algorithm leads to the wrong choice of $k$
\emph{more than half of the times}. In fact, when $t \leq 256$, the
correct value of $k$ is identified no more than $7 \%$ of the times,
which is less than expected if a random choice among the considered
values of $k$ were made. The results clearly improve with $t=1024$, but
still the wrong value of $k$ is identified at least $15\%$ of the
times, and close to $50\%$ of the times in most cases. These results
provide further evidence of the superiority of PPS sampling.

As mentioned in the main paper, we also conducted an experiment to
assess the accuracy of the estimate returned by
\textsc{PPS-Silhouette} under squared Euclidean distances (which do
not satisfy the triangle inequality). Thanks to the availability of
the efficient Apache Spark routine that computes the exact value with
subquadratic work, for this experiment we were able to consider a
large instance of the synthetic dataset with $n=10^8$ points, together
with two clustering granularities ($k=5, 10$). We ran
\textsc{PPS-Silhouette} with expected sample size $t=64$.  For $k=5$,
\textsc{PPS-Silhouette} provides estimates with an average error
smaller than $0.005$ and a maximum error of $0.008$, while for $k=10$
the average error is smaller than $0.058$ and the maximum error is
$0.069$.  These results show that our \textsc{PPS-Silhouette}
algorithm can be used to provide accurate estimates of the silhouette
with a limited number of samples even for squared Euclidean distances.

Overall, the results show that \textsc{PPS-Silhouette} exhibits a high
level of accuracy even for relatively small values of the expected
sample size $t$. Since the total work performed by
\textsc{PPS-Silhouette} is $O(nkt)$, this implies that reliable
estimates can be obtained in much less than the $\Theta(n^2)$ work
required by the exact computation.

\subsection{Performance Evaluation}
\label{sec:scalab}
In this section, we report the results of several experiments aimed at comparing the distributed performance of the
MapReduce implementation of \textsc{PPS-Silhouette} against the one of
MapReduce implementations of the uniform sampling and the
Frahling-Sohler optimization baselines, and of the optimized Spark
routine for squared distances, on very
large datasets.

In the first experiment, we assessed the scalability of the MapReduce version of
\textsc{PPS-Silhouette}.  For this
experiment, we used four instances of the synthetic dataset (see 
Section~\ref{sec:data}) with $n=10^6, 10^7, 10^8, 10^9$ elements and $k=5$
clusters. We ran
MapReduce implementation of \textsc{PPS-Silhouette} with $t=64$ using
$w=1,2,4,8,16$ workers (each on a single core).  For the
dataset with $n=10^9$ elements, we considered only $w=4, 8, 16$ since for
lower levels of parallelism the local memory available to each core
was not sufficient to process the $n/w$ elements assigned to each
worker.   The results are shown in Figure~\ref{fig:scalab}. We observe that
for the two largest datasets \textsc{PPS-Silhouette} exhibits linear
scalability. For $n=10^7$, the algorithm still exhibits linear
scalability for up to $8$ workers, while there is a limited gain in
using 16 workers. For $n=10^6$, scalability is still
noticeable, however, due to the relatively small size of the dataset, the
behavior is less regular due to the stronger incidence of
communication overheads and caching effects.  Also, for a fixed number
of workers, we observe that the algorithm scales linearly with
the number of elements, in accordance with our theoretical analysis.

In successive experiments, we compared the running times
of the MapReduce implementations of \textsc{PPS-Silhouette} (with $t=64$) and of the baselines.
We compared  \textsc{PPS-Silhouette} with uniform sampling algorithm and times showed that, as
expected, uniform sampling is slightly faster (by $10-20\%$ in all
experiments), due to the absence of the probability calculation
phase. This mild advantage of uniform sampling is however
counterbalanced by its much lower level of accuracy, quantified in the
previous section, for equal sample size. We compared \textsc{PPS-Silhouette} with
the exact computation based on the definition (\textsc{DEF}), and of the exact computation based on the   
Frahling and Sohler optimization (\textsc{FS}), with fixed, maximum degree of parallelism
$w=16$, and using the smallest dataset of the previous experiment, namely the one
with $n=10^6$ elements. While \textsc{DEF} was stopped after 6000 seconds,  \textsc{FS}
completed the execution in 1135 seconds, and \textsc{PPS-Silhouette} in 14 seconds, thus
(more than) two orders of magnitude faster than \textsc{DEF} and \textsc{FS}. Also,
we point out that \textsc{PPS-Silhouette} was able to estimate the silhouette for 
a three orders of magnitude larger dataset ($n=10^9$) in 2433 seconds, which is about twice 
the time required by \textsc{FS} on the dataset with $10^6$ elements. Considering the
accuracy of the estimation provided by \textsc{PPS-Silhouette}, assessed in the previous subsection, 
these comparisons provide evidence of its practicality.



\begin{figure}[h!]
\centering
\includegraphics[width=0.4\textwidth]{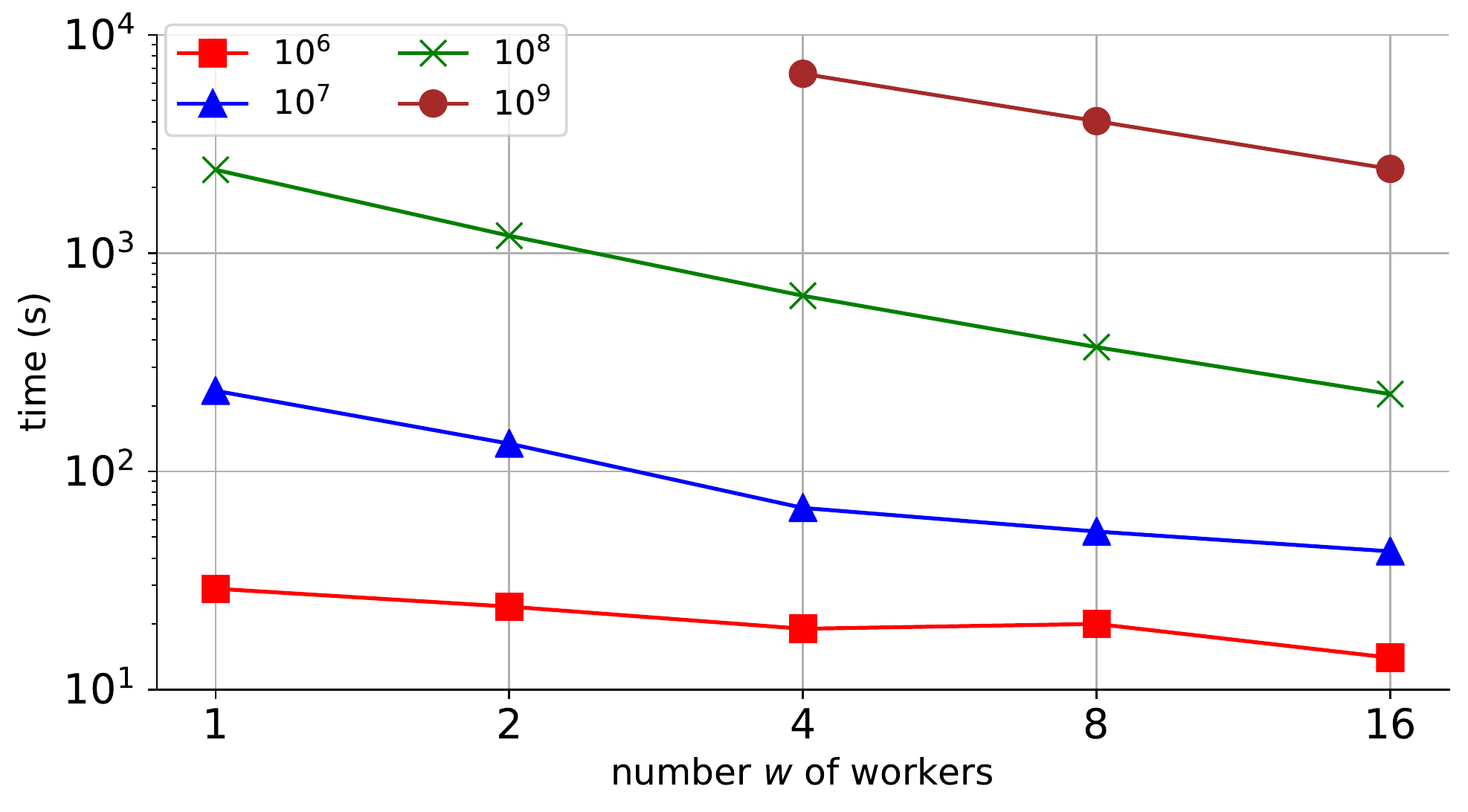}
\caption{Running time (median over 5 runs) of our \textsc{PPS-Silhouette} algorithm as function of the number $w$ of workers, for datasets of different sizes.}
\label{fig:scalab}
\end{figure}

Finally, we compared the performance of \textsc{PPS-Silhouette} under
squared Euclidean distances with the one of the optimized exact method
available in Apache Spark for such distances, using our synthetic
dataset with $n=10^8$ elements and two values of $k$, namely $k=5,
10$. Since our goal in this experiment was to compare the best
performance achievable by the two implementations, and since our
\textsc{PPS-Silhouette} algorithm provides very accurate estimates
already with a limited number $t$ of samples, we fixed $t=32$ and
$w=16$ workers.  For both values of $k$ the estimates from
\textsc{PPS-Silhouette} are very precise (average error $0.007$ for
$k=5$ and $0.087$ for $k=10$), while the running time of
\textsc{PPS-Silhouette} is comparable to, even if higher than (up to
approximately two times), the running time of the optimized Apache
Spark method.


%

\section{Conclusions}
\label{sec:conclusions}

In this work, we developed the first efficient, sampling-based
algorithm to estimate the silhouette coefficient for a give clustering
with provable approximation guarantees, whose running time
dramatically improves over the quadratic complexity required, in the
general case, by the exact computation.  We provided a distributed
implementation of our algorithm in Map-Reduce which runs in constant
rounds, and requires only sublinear space at each worker.  The
experimental evaluation conducted on real and synthetic datasets,
demonstrates that our algorithm enables an accurate estimation of the
silhouette for clusterings of massive datasets for which the exact
computation is out of reach. Future research should target a more
extensive experimentation on real-world datasets.  Also, it would be
interesting to devise a more flexible, dynamic version of the algorithm
where the desired accuracy can be incrementally refined by reusing
(part of) the previously sampled points.

\section{Acknowledgments.}
This work was supported, in part, by MIUR, the Italian Ministry of
Education, University and Research, under PRIN Project n. 20174LF3T8
AHeAD (Efficient Algorithms for HArnessing Networked Data), and grant
L. 232 (Dipartimenti di Eccellenza), and by the University of Padova
under projects ``STARS 2017: Algorithms for Inferential Data Mining"
and ``SID 2020: Resource-Allocation Tradeoffs for Dynamic and Extreme
Data''.  The authors want to thank Matteo Ceccarello for his
invaluable help with the Apache Spark framework.

\end{document}